\pdfoutput=1
\newif\ifFull
\Fulltrue
\ifFull
\documentclass[11pt]{article}
\usepackage{fullpage}
\else
\documentclass{sig-alternate}
\fi
\usepackage{graphicx}
\usepackage{cite}
\usepackage{url}
\usepackage{latexsym}
\usepackage{amssymb}
\usepackage{amsmath}
\newcommand{\R}{{\bf R}}
\ifFull
\newenvironment{proof}{\noindent{\bf Proof:}}{\hspace*{\fill}\rule{6pt}{6pt}\bigskip}
\else
\conferenceinfo{ACM GIS}{'08, November 5-7, 2008. Irvine, CA, USA}
\crdata{978-1-60558-323-5/08/11}
\CopyrightYear{2008}
\fi

\newtheorem{theorem}{Theorem}
\newtheorem{lemma}[theorem]{Lemma}

\newtheorem{corollary}[theorem]{Corollary}

\begin{document}

\title{Studying Geometric Graph Properties of Road Networks \\
       Through an Algorithmic Lens}

\author{%
David Eppstein \\
Dept.~of Computer Science \\
University of California, Irvine \\
\url{http://www.ics.uci.edu/~eppstein}
\and
Michael T. Goodrich \\
Dept.~of Computer Science \\
University of California, Irvine \\
\url{http://www.ics.uci.edu/~goodrich}
}

\date{}
\maketitle 

\begin{abstract}
This paper studies the geometric graph properties of
road networks from an algorithmic
perspective, focusing on empirical studies that yield useful
properties of road networks and how these properties
can be exploited in the design of fast algorithms.
Unlike previous approaches, our study is
not based on the assumption that road networks are planar graphs.
Indeed, based on experiments we have performed on
the road networks of the 50 United States and District of Columbia,
we provide strong empirical evidence that
road networks are quite non-planar.
Our approach instead is directed at finding
algorithmically-motivated properties of road networks as 
non-planar geometric graphs, focusing on alternative
properties of road networks that can still lead to 
efficient algorithms for such problems as shortest paths and Voronoi
diagrams.
In particular, we study road networks as
\emph{multiscale-dispersed} graphs, which is a concept we formalize in
terms of disk neighborhood systems.
This approach allows us to develop fast algorithms for
road networks without making any
additional assumptions about the distribution of edge weights.

\end{abstract}

\section{Introduction}
With the advent of online mapping systems, including proprietary
systems like Google Maps and
Mapquest, and open-source projects
like \url{http://wiki.openstreetmap.org/},
there is an increased interest in studying
road networks as natural artifacts.
Indeed, the Ninth DIMACS Implementation Challenge\footnote{See
http://dimacs.rutgers.edu/Workshops/Challenge9/}, from 2006, was
dedicated to the algorithmic study 
of road networks from the perspective of methods for
solving shortest path problems on road networks.
(See Figure~\ref{fig-anchorage} for an example extracted from this data.)

\begin{figure}[tb]
\vspace*{-1.15in}
\begin{center}
\includegraphics[width=3.5in]{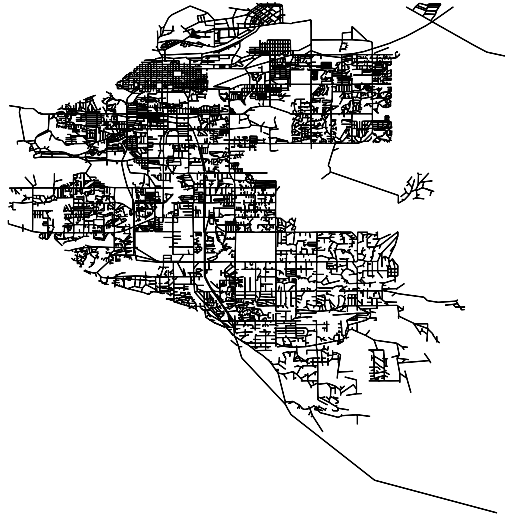} 
\end{center}
\vspace*{-0.35in}
\caption{A portion of the road network 
from the U.S.~Census 2000 TIGER/Line files. 
This example illustrates the
road network for 
Anchorage, Alaska, which consists of over 9,000 vertices and 23,000 edges.
}
\label{fig-anchorage}
\end{figure}

Road networks provide an interesting domain of study, from
an algorithmic perspective,
as they combine geographic and graph-theoretic information in one structure.
In particular, road networks can be viewed as instances of 
\emph{geometric graphs}~\cite{ttgg}, 
that is, graphs in which each vertex 
is associated with a unique point in $\R^2$ and each edge is associated with
a simple curve joining the points associated with its end vertices.
In the case of a road network, we
create a vertex for every road intersection or major jog and
we create an edge for every road segment that joins two such points.
In addition, road networks contain geographic information, in that 
vertices are usually labeled with
their (GPS) longitude and latitude coordinates and edges are labeled
with their length.
These basic data attributes of road networks, together
with the fundamental anthropological nature of the way these
networks encapsulate the way that societies organize themselves,
make road networks unique data objects.
Therefore, following similar studies of natural 
artifacts through the \emph{algorithmic lens}
(e.g., see~\cite{ai-cs2cr-07,ac-itg-05,c-aims-06,%
c-cyihg-06,c-rmgta-06,p-alhcp-07,w-ct-06}),
we are interested in this paper in the study of algorithmically
motivated properties of road networks as geometric graphs.

\subsection{Geometric Graph Properties of Road Networks}
Because roads are built on the surface of the Earth, which as
a first-order approximation is a sphere,
it is tempting to view them as 
plane graphs~\cite{dett-gd-99,eg-sbdgt-70,t-pg-93}, 
that is, graphs that are drawn
on the surface of a sphere without edge crossings.
Indeed, some statistics commonly used
to characterize road networks only make sense for planar graphs
(e.g., see~\cite{c-esag-96}).
For example, the \emph{Alpha index} for a road network 
with $n$ vertices is the ratio of the
number of internal fundamental cycles to $2n-5$, which is 
called ``the maximum
possible,'' but this maximum doesn't hold for non-planar graphs.
Likewise, the \emph{Gamma index} is the ratio of the
number of edges in a road network to $3n-6$, which is another
``maximum possible'' that doesn't hold for non-planar graphs. (The
values $2n-5$ and $3n-6$ come from Euler's formula for planar graphs.)
Such statistics
are often used to measure the density of road networks, and,
as standardized numbers, they probably suffice for this purpose, but
they may also give the false impression that road networks are plane
graphs.

Anyone who has seen an expressway underpass has
firsthand experience of a counter-example to the planarity of road networks.
Of course, if such edge crossings were rare, it
would probably still be okay to think of road networks as plane graphs or
``almost plane'' graphs. 
In fact, as we show in the experimental
results of this paper, real-world road networks can have many edge
crossings. For example, the road network for California---with its
extensive freeway system---has almost 6,000 edge crossings!
Thus, the main goal of this paper is to exploit an alternative 
geometric characterization
of road networks that reflects their non-planar nature. 
Additionally, we wish to avoid making any
assumptions about the distribution of edge weights, because it is 
important in many transportation planning problems to use artificial 
weights that do not reflect the geometry of the input 
and that may vary from user to user~\cite{Epp-SJC-03}.

The approach we study in this paper is based on characterizing
road networks as subgraphs of disk intersection graphs,
in a way that takes
advantage of their bounded-depth fractal nature~\cite{bl-fc-94}.
In particular, we exploit the manner in which road networks reflect
the way populations naturally create geographic features
that are well-separated at multiple scales in a self-similar fashion.
We introduce a formalism that characterizes road networks as
\emph{multiscale-dispersed} graphs, which is itself based on viewing road
networks as subgraphs of disk intersection graphs.

A \emph{disk intersection graph} is defined from a collection $S$ of disks 
in $\R^2$ by creating a vertex for every disk in $S$ and defining an edge 
for every pair of intersecting disks in $S$.
Disk intersection graphs include the planar graphs,
since, by a beautiful result of Koebe~\cite{k-kdka-36}, 
every planar graph can be
represented as the intersection graph of a set of disks that
intersect only at their boundaries. 
Moreover, 
by a result of Mohar~\cite{m-ptcpa-93},
such embeddings can be constructed in polynomial time.
A set $S$ of disks in the plane 
defines a \emph{$k$-ply} neighborhood system if, for
any point $p\in\R^2$, the number of disks from $S$ containing $p$
is at most $k$.
Thus, the collection of ``kissing'' disks used in 
the representation of Koebe~\cite{k-kdka-36} 
for a planar graph defines a $2$-ply neighborhood system.

The natural way to define a 
disk neighborhood system from a road network $G$ is to center a disk at
each vertex $v$ in $G$ and define its radius
to be half the length of the longest edge incident to $v$ in $G$.
The intersection graph of such a disk neighborhood system 
is guaranteed to contain $G$ as a subgraph.
Thus, we define this as the \emph{natural} disk neighborhood system
determined by the road network $G$, and we use this notion throughout
this paper.
(See Figure~\ref{fig-natural}.)

\begin{figure}[htb]
\vspace*{-1.0in}
\begin{center}
\includegraphics[width=3.75in]{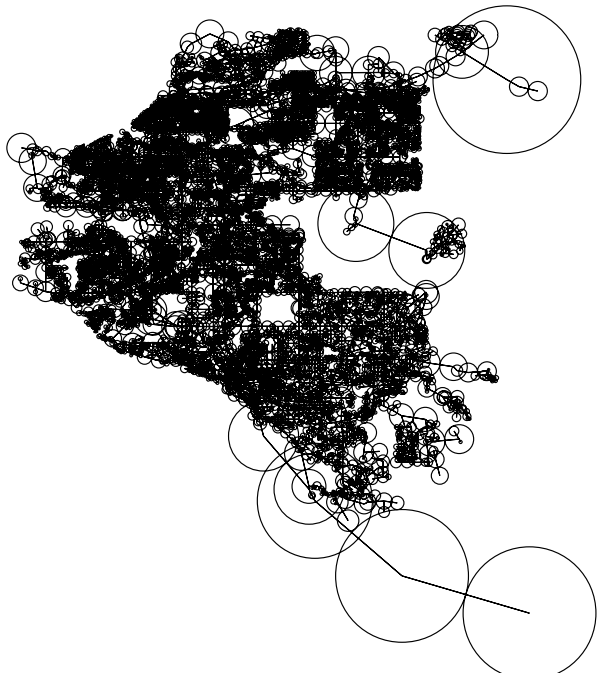} 
\end{center}
\vspace*{-0.15in}
\caption{
The natural disk neighborhood system for the road
network of Anchorage, Alaska, which was presented in
Figure~\protect{\ref{fig-anchorage}}.}
\label{fig-natural}
\end{figure}

We exploit this characterization by showing that it allows us to
utilize a technique used
in many fast algorithms for planar graphs---the
the existence of efficient
methods for finding small-sized separators~\cite{lt-stpg-79}.
These are small sets of vertices whose removal separates the graph into two
subgraphs of roughly equal size.
Interestingly, like planar graphs, $k$-ply
neighborhood systems of disks have small-sized separators 
(e.g., see~\cite{ar-dsa-93,MilTenThu-SJSC-95,mttv-sspnn-97,st-dpps-96}). 
In order to handle the intricacies of road networks, the natural disk 
neighborhood systems of which may not be $k$-ply, we generalize these 
results by showing that small separators exist for a larger 
related class of graphs. 
Moreover, there exist algorithms for finding good separators 
in $k$-ply neighborhood
systems and our generalizations of them, and these algorithms 
can be used to construct
efficient algorithms for such graphs, as we show.

Additional experiments we provide in this paper 
justify the claim that real-world 
road networks have the properties
mentioned above.
Moreover,
we show that our approach
leads to improved algorithms and data structures for dealing
with real-world road networks, including shortest paths and Voronoi
diagrams.
In general, we desire comparison-based
algorithms that require no additional 
assumptions regarding the distribution of edge weights, so that our algorithms
can apply to a wide variety of possible edge weights, including
(non-metric) combinations of distance, 
travel time, toll charges, and subjective scores rating safety
and scenic interest.

\begin{figure*}[tb]
\begin{center}
\includegraphics[width=1.9in]{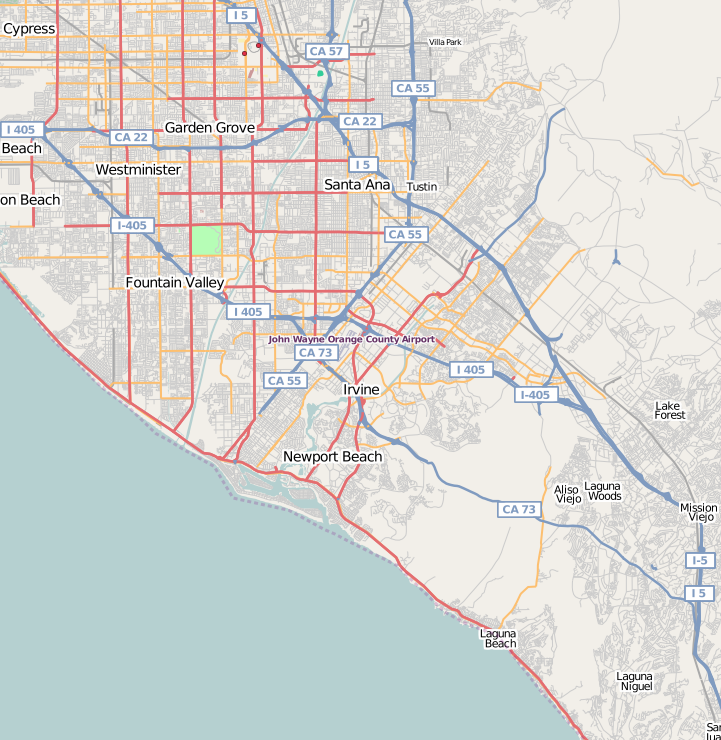} 
\hspace*{0.3in}
\includegraphics[width=1.9in]{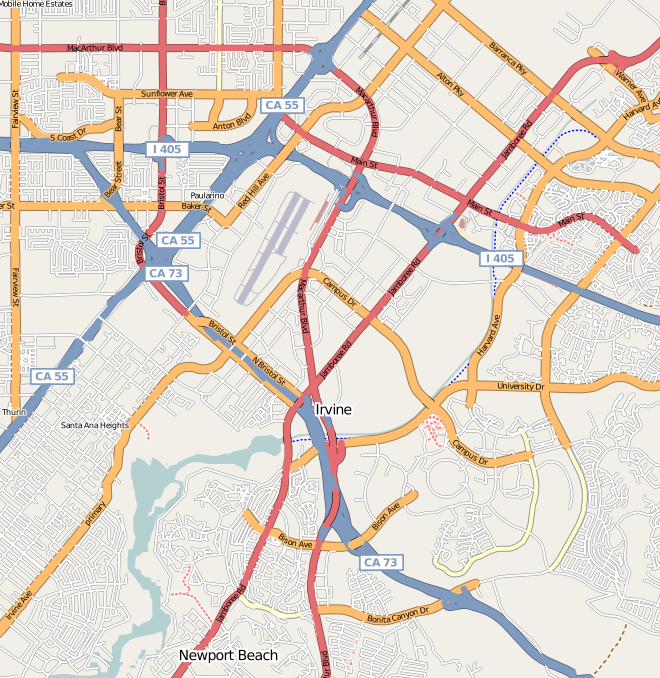} 
\hspace*{0.3in}
\includegraphics[width=1.9in]{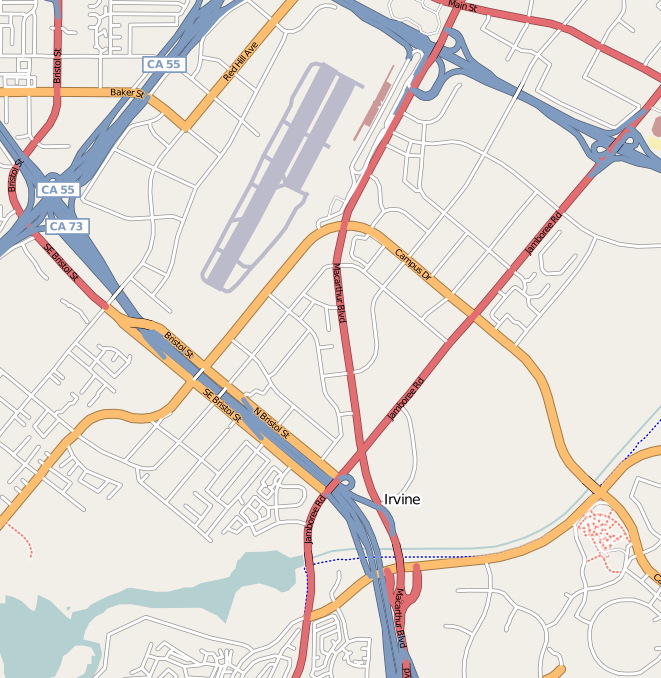} 
\end{center}
\vspace*{-10pt}
\caption{
The multiscale dispersion of real-world maps. 
These three images show 
various scales of maps of Irvine, California.
Images are from
\textsf{http://wiki.openstreetmap.org/},
under the Creative Commons attribution-share alike license.
}
\label{fig-fractal}
\end{figure*}

\subsection{Previous Related Work}
There has been considerable work in the transportation literature on
studies of the properties of road networks (e.g., see~\cite{c-esag-96}),
including analyses based on the Alpha and Gamma indices. As
mentioned above, however, much of this work implicitly assumes that
road networks are planar.

In the algorithms community, there has 
been considerable prior work on shortest
path algorithms for road networks 
and Euclidean graphs 
(e.g., see~\cite{gh-cspasm-05,hsww-csuts-05,%
kp-sppls-06,ss-hhhes-05,sv-speg-86,zn-spaeu-98}
and the program for the 
Ninth DIMACS Implementation Challenge).
This prior work on road network algorithms takes a decidedly
different approach than we take in this paper, however, in that this
prior work focuses on using special properties of the edge weights
that do not hold in the comparison model,
whereas we study road networks as a graph family and desire
properties that would result in efficient
algorithms in the comparison model.

One of the main ingredients
we use in our algorithms is the existence of small separators in
certain graph families (e.g., see~\cite{lt-stpg-79,m-fsscs-86}).
Previous work on separators includes the seminal contribution of
Lipton and Tarjan~\cite{lt-stpg-79}, who showed that $O(\sqrt{n})$-sized
separators exist for $n$-vertex planar graphs and these can be computed
in $O(n)$ time.
Goodrich~\cite{g-psppt-95} shows that recursive
$O(\sqrt{n})$-separator decompositions can be constructed for planar
graphs in $O(n)$ time.
A related concept is that of geometric separators, which use
geometric objects to define separators in geometric graphs
(e.g., see~\cite{ar-dsa-93,MilTenThu-SJSC-95,mttv-sspnn-97,st-dpps-96}), 
for which Eppstein
{\it et al.}~\cite{emt-dltag-93} provide a linear-time construction
algorithm, which translates into an $O(n\log n)$-time recursive separator
decomposition algorithm.

The specific algorithmic problems we study are
well-known in the general algorithms and computational geometry
literatures.
For general graphs with $n$ vertices and $m$ edges, 
excellent work can be found on 
efficient algorithms in the comparison model,
including single-source
shortest paths~\cite{clrs-ia-01,gt-adfai-02,r-rrsss-97}, 
which takes $O(n\log n + m)$
time~\cite{ft-fhtui-87}, and
Voronoi diagrams~\cite{a-vdsfg-91,ak-vd-00},
whose graph-theoretic version
can be constructed in $O(n\log n + m)$ time~\cite{e-tgvda-00,m-afaas-88}.
Note that none of these algorithms run in linear time, even for
planar graphs. Nevertheless, linear-time algorithms for planar graphs
are known for single-source shortest paths~\cite{hkrs-fspap-97},
which unfortunately do not immediately
translate into linear-time algorithms for road networks.
In addition, there are a number of efficient
shortest-path algorithms that make assumptions about 
edge weights~\cite{g-saspp-93,gh-cspasm-05,m-ssspa-01,t-usspp-99};
hence, are not applicable in the comparison model.
Eppstein {\it et al.}~\cite{egs-ltagg-09} show how to find all the
edge interesections in an $n$-vertex straight-line geometric graph
in $O(n + k\log^{(c)} n)$ expected time, where $k$
is the number of pairwise edge intersections and $c$ is any fixed
constant, and they show how this can be combined with the method of
Goodrich~\cite{g-psppt-95} to construct a recursive
$O(\sqrt{n})$-separator decomposition for such a graph
in $O(n)$ time if $k$ is $O(n/\log^{(c)} n)$.
There method does not result in a geometric separator decomposition,
however, as our approach in this paper does.

\subsection{Our Results}
In this paper, we study properties of road networks
that can be exploited in efficient algorithms for such networks.
In particular, we view road networks 
as a special class of non-planar geometric graphs and we show 
how to design efficient algorithms for them that are
based on their characterization as embedded graphs that
that have well-dispersed vertices and edges.
We use the automatic multiscale nature of 
disk neighborhoods to capture 
the property that 
the spatial distribution of edges is similar at multiple scales, 
like fractals, except that this recursive self-similarity 
has a bounded depth~\cite{bl-fc-94}.
(See Figure~\ref{fig-fractal}.)

We provide experimental evidence that debunks the 
belief that road networks are plane
graphs or even ``almost planar.'' Instead, we provide
empirical evidence that real-world road networks have natural disk
neighborhood systems that are subgraphs of
$k$-ply disk intersection graphs (with a small number of exceptions), 
for constant $k$. 
That is, road networks are multiscale-dispersed graphs.
Our analysis uses the U.S.~TIGER/Line road network database, as
provided by the Ninth DIMACS Implementation Challenge,
which is comprised
of over 24 million vertices and 29 million edges.

Viewing road networks as multiscale-dispersed 
graphs allows us to prove that these networks 
have small geometric separators, which can be found quickly.
Furthermore, we show how to use recursive separator decompositions for road
networks in the design of fast
shortest path and Voronoi diagram
algorithms for road networks.
We also provide a fast algorithm for constructing the arrangement of
circles in a natural disk neighborhood system, using additional
algorithmic properties of road networks, which we justify
empirically.
Thus, we justify the claim that viewing road networks through the
algorithmic lens can lead to the discovery of properties of these
networks that result in algorithms that run faster than algorithms
that would simply view these networks as standard graphs.

\begin{figure*}[htb]
\vspace*{-0.3in}
\includegraphics[width=6in]{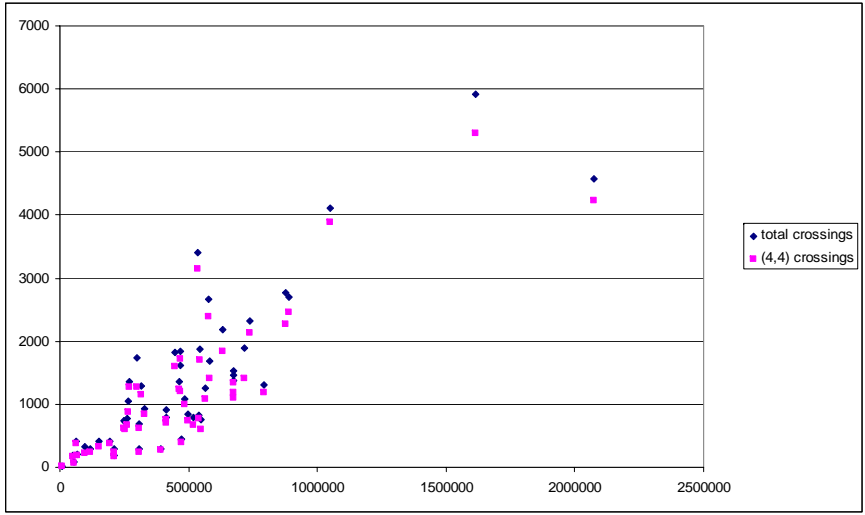}
\vspace*{-1.34in}
\caption{
A scatter plot of the number of edge crossings for TIGER/Line road
networks, indexed by network size, listed as both total number of
crossings and number of crossings between pairs of edges such that
both are in the highest level of the road network hierarchy.
}
\label{fig-cross}
\end{figure*}

\section{Non-Planarity}

As noted above, every road
network is a graph $G=(V,E)$, having a set, $V$, of $n$ vertices, 
defined by road junctions
and on/off ramps, and a set, $E$,
of $m$ edges, defined by the uninterrupted
pieces of roads, highways, and expressways between such vertices.
Since the number of roads that meet at a single junction cannot
be arbitrarily large, road networks have \emph{bounded degree}. Thus,
the number, $m$, of edges in a road network is $O(n)$.
For this reason, we will focus on $n$ as the key factor defining
the ``size'' of a road network.

As mentioned above, road networks are \emph{geometric graphs}~\cite{ttgg}. 
That is, they are ``drawn'' on a surface, so that 
vertices are placed at distinct points and edges correspond to simple
curves, which in this paper we assume are straight line segments.
In addition, the edges incident to each edge
are given in a specific order, e.g., clockwise or counter-clockwise.
In addition to this ordering information, road networks also have
geographic information: namely, each vertex in a road network is given 
two-dimensional GPS coordinates, which in turn imply the line
segments defined by edges.
If none of these segments cross each other, and the surface is a sphere or a plane, 
then the embedding is said to 
be a \emph{plane graph}~\cite{dett-gd-99,eg-sbdgt-70,t-pg-93}.
Graphs that can be drawn 
as plane graphs are called \emph{planar graphs}.

As mentioned above, much previous work has been done that 
studies road networks as
plane graphs. The first set of experimental results we share 
in this paper show that in the road networks of the 50 
United States and District
of Columbia are not plane graphs.
Indeed, these experiments show that none of these graphs are even
``almost planar'' graphs. In actuality, however, every road network
in the TIGER/Line database contains multiple edge crossings, as shown in
the plot of Figure~\ref{fig-cross}.
Indeed, some road networks have thousands of crossings.
In general, the road networks tested have a number 
of crossings that appears to be
proportional to $\sqrt{n}$, where $n$ is the number of vertices.

There is a refinement of the belief that road-networks are plane graphs, 
which deals with the hierarchal nature of road networks. In particular,
the edges in most road-network databases are categorized
according to a discrete hierarchy, with European road
networks usually having 13 types of road segments and U.S.~networks
usually having 4, which correspond roughly to 
U.S.~highways, state highways, major roads, and local roads.
A refined belief could be that road networks are planar at each level of
hierarchy, given that these hierarchies are intended to capture the
various scales of ``resolution'' in road networks.
Unfortunately, this belief is also unsupported by the data, as
a refined analysis of the crossing counting data shows that most
of the crossings are actually between edges in the same level of the
hierarchy.
Indeed, Figure~\ref{fig-cross}
shows that the vast majority of crossings
are between pairs of edges in the highest level of the
hierarchy (level 4), which is intended for local roads.

\begin{figure*}[htb]
\vspace*{-0.6in}
\includegraphics[width=6in]{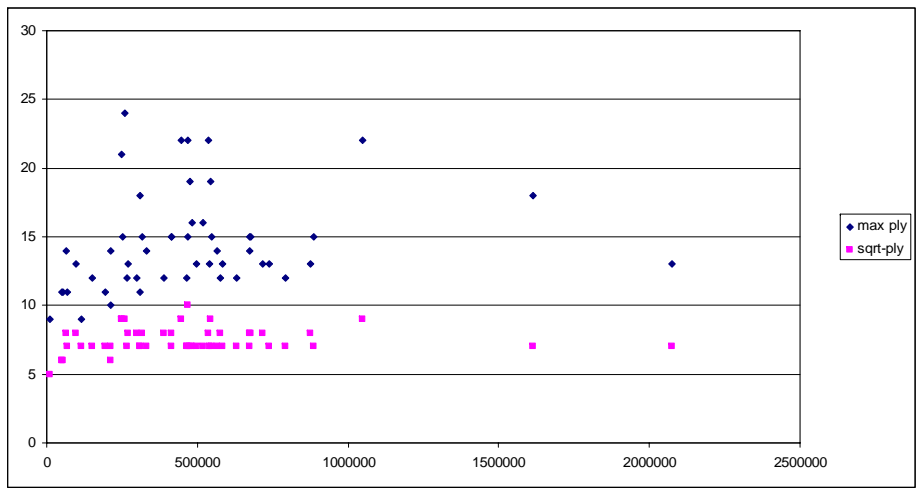} 
\vspace*{-1.5in}
\caption{
A scatter plot of the maximum ply of vertices in each road network and a
scatter plot of the $\lfloor \sqrt{n}\rfloor$-th largest vertex ply, both
indexed by network size.}
\label{fig-ply}
\vspace*{-4pt}
\end{figure*}

The take-away message from the above analysis is that if we view
road networks
as being geometric graphs drawn on a sphere, then they have many edge crossings.
In particular, if we have an $n$-vertex road network, then there will 
typically be
$\Theta(\sqrt{n})$ edge crossings when all those segments are
projected on the surface of a sphere using their GPS coordinates.
Of course, road segments are not projected on the surface of a sphere and
they don't actually intersect; hence,
roads that have apparent crossings are the result of 
overpasses, tunnels, data errors, or representational imprecisions. 

A natural question, of course, is to ask why the number of edge
crossings in road networks is proportional to $\sqrt{n}$.
One possible explanation comes by way of analogy.
Imagine an idealized city, which we will call Gotham City,
that has a road network that, at its finest grain of detail,
consists of an orthogonal grid of
roads, like downtown Manhattan. 
Suppose at the next coarser grain of detail, Gotham City has an
expressway network built on top of its fine-grain road grid, with the
number of expressways being constant.
Since expressways are built for fast transport, they are unlikely to
have large zig-zags or consist of large spiral segments. Instead,
they would be built to cut straight through
Gotham City. Since any straight segment cutting through an orthogonal
grid will intersect $O(\sqrt{n})$, this implies that the number of
edge crossings in Gotham City's road network is $O(\sqrt{n})$. 
Likewise, a random Delaunay triangulation  (which is a natural way of
triangulating a set of points), defined on $n$ points,
will have a similar property---in that any line
will intersect an expected $O(\sqrt{n})$ of segments of such a
triangulation~\cite{bd-snrdt-07}.
A similar phenomenon, occurring at a grander scale,
could be the motivating reason why real-world road networks tend to
have $O(\sqrt{n})$ edge crossings, with
the existence of
noisy data (which tends to increase the number of edge
crossings) offsetting the fact that expressways and large highways are
built in a way that tries to avoid needless overpasses (which tends
to reduce the number of edge crossings).
In any case, 
we can use the non-planarity of road networks
to motivate an alternative characterization of road networks, as seen
through an algorithmic lens.

\section{Disk Neighborhood Systems}
We define a neighborhood system $S$ of $n$ disks in the plane to be 
an \emph{$f(n)$-exceptional} $k$-ply system if there is a function,
$f(n)$, such that there exists a subset 
$T$ of $f(n)$ disks in $S$ so that $S-T$ 
forms a $k$-ply neighborhood system and every disk in $T$
intersects at most $f(n)$ disks in $S$.
Intuitively, the disks in $T$ are the ``exceptions that prove the
rule'' that $S$ is, for the most part, a $k$-ply neighborhood system.
Given a graph $G$ with $n$ vertices that is embedded in $\R^2$,
we say that $G$ is a \emph{multiscale-dispersed} graph if its
natural disk neighborhood system is
an $O(\sqrt{n})$-exceptional $O(1)$-ply system\footnote{
  Note that the disks in the exception set $T$ are not necessarily
  identified at the outset---we are simply saying that $T$
  exists in order for $G$ to be a multiscale-dispersed graph.}.
As mentioned above,
the motivation for the constant-ply part of 
this definition is that the edges and vertices of
road networks tend to be well-spaced, in a bounded-depth fractal
way~\cite{bl-fc-94}. By including disks of widely varying sizes, our definition allows the simultaneous treatment of road network features of widely varying scales, but (with the removal of the exceptional disks) the remaining disks must be well dispersed in order to achieve the $k$-ply property.

The choice of an $O(\sqrt{n})$ asymptotic growth rate for the exceptional part of the definition 
comes from the possible existence of overlapping large disks.
Such large disks are troublesome with respect to our desire for
having all disks define a constant-ply
disk neighborhood system if 
these large disks all overlap a dense urban region.
Of course, large disks in a natural disk neighborhood system
are defined by the long roads.
Moreover, it is not unusual for long roads to lead to dense urban
regions, as in
the classic ``all roads lead to Rome'' phenomenon.
Nevertheless,
it is natural for urban areas to have shapes with bounded
aspect ratios
(e.g., see~\cite{bl-fc-94}),
since inhabitants of such regions 
naturally like to be close to each other.
Thus, the map footprints of urban areas tend to have 
perimeters that are proportional to the square root of their areas.
Thus, if a dense city, like Gotham City, has $O(n)$ internal
edges, it will have $O(\sqrt{n})$ boundary edges. But this
necessarily bounds the number of long incoming roads to be
$O(\sqrt{n})$.

Our definitions are not simply motivated by geographical arguments,
however.
In Figure~\ref{fig-ply} we show a scatter plot from experiments that 
support the claim that the maximum
ply of natural disk neighborhood systems for real-world $n$-vertex road
networks is higher than the
$\lfloor \sqrt{n}\rfloor$-th largest vertex ply, which itself is a small
constant.
In particular, even for small road networks, it shows
that the largest circle-center ply 
can be more than $20$, whereas all the 
$\lfloor \sqrt{n}\rfloor$-th largest vertex plies are at most $10$.
Using an argument similar to that used in the proof
of the following lemma (\ref{lemma-pairs}), 
we can show that the maximum ply of a disk
neighborhood system is within a constant factor of the maximum ply of the
centers of the disks, that is, the vertices of our road networks.
Thus, this empirical analysis supports the claim that 
the natural disk neighborhood systems of
real-world road networks are $O(\sqrt{n})$-exceptional $O(1)$-ply systems.
We explore one implication of the multiscale-dispersed graph
property in the following lemma.

\begin{lemma}
\label{lemma-pairs}
If $G$ is an $n$-vertex multiscale-dispersed graph,
then the number of pairs of intersecting disks in $G$'s natural disk
neighborhood system is $O(n)$.
\end{lemma}
\begin{proof}
Since $G$ is a multiscale-dispersed graph,
its natural disk neighborhood system, $S$, is an $O(\sqrt{n})$-exceptional
$O(1)$-ply neighborhood system.
Let $T$ denote the subset of $S$ of at most $O(\sqrt{n})$ disks in 
$S$ such that each disk in $T$ intersects at most $O(\sqrt{n})$ other disks
in $S$.
Clearly, the number of pairs of intersecting disks in $S$ such that one of
them is in $T$ is $O(n)$.
Thus, we have only to count the number of pairs of intersecting disks such
that both are in $S-T$, which we could do using a lemma (3.3.2) of
Miller {\it et al.}~\cite{mttv-sspnn-97},
but we include a counting argument here for completeness and because it leads to
a better constant factor.
Let $k\in O(1)$ denote the ply of the disks in $S-T$,
and let us consider the two cases for a pair of intersecting disks 
$c$ and $d$ in $S-T$:
\begin{enumerate}
\item
One of $c$ and $d$ contains the center of the other.
In this case, we charge this intersection to a disk in $\{c,d\}$ whose center
is contained inside the other disk.
Let us call this a ``containment'' charge.
\item
Neither $c$ nor $d$ contains the center of the other.
In this case, we charge this intersection to the smaller of the disks
$c$ and $d$.
Let us call this a ``tall'' charge.
\end{enumerate}
To complete the proof, we need only account for these two types of 
charges in turn.
First, note that,
since $S-T$ is a $k$-ply system, each disk can receive at most 
$k\in O(1)$ containment charges (for at most $k$ disks can 
contain the center of any disk).
So we have only to
account for the number of tall charges a disk $c$ can receive.
Suppose $d$ is a disk larger than $c$ that intersects $c$ 
and does not contain $c$'s center.
There are at most $k$ other disks that can contain $d$'s center and also
intersect $c$. Thus, the total number of tall charges $c$ can receive is at
most $k$ times the number of disks larger than $c$ 
that do not contain each other's centers 
and which intersect $c$.
If we order the centers of these disks around $c$, then each consecutive pair
of centers defines a triangle whose side oppose $c$ is largest; hence, its
angle at $c$ is largest. Therefore, there can be at most $6$ such disks,
which implies that the number of tall charges $c$ can receive is at most
$6k$, which is $O(1)$.
\end{proof}

\begin{corollary}
\label{cor-arrangement}
The combinatorial complexity of the arrangement, ${\cal A}(S)$, of circles
determined by the natural disk neighborhood system, $S$, of a 
multiscale-dispersed graph is $O(n)$.
\end{corollary}
\begin{proof}
Every intersecting pair of disks in $S$ determines at most two vertices in
${\cal A}(S)$.
Thus, by Lemma~\ref{lemma-pairs}, the combinatorial complexity 
of ${\cal A}(S)$ is $O(n)$.
\end{proof}

We have shown empirically above that an $n$-vertex 
road network is a geometric graph having $\Theta(\sqrt{n})$ edge crossings.
Interestingly, our characterization of road networks as
bounded-degree multiscale-dispersed graphs can be used to 
imply an alternative upper bound on the number of
crossings among the edges of a road network.
This linear upper bound is not as strong as the above characterization,
but it is nevertheless much better than the quadratic bound that would be 
expected from a random geometric graph.

\begin{lemma}
\label{lemma-crossings}
Let $G$ be an $n$-vertex multiscale-dispersed graph
with bounded degree. Then the number of 
pairs of edges of $G$ that cross in $\R^2$ is $O(n)$.
\end{lemma}
\begin{proof}
Let $S$ be the natural disk neighborhood system for $G$.
Suppose two edges $e_1=(u,v)$ and $e_2=(w,z)$ in $G$ cross in $\R^2$ 
at some point $p$ (note: we do not assume these edges are straight lines).
Let $d_{e_x}(p,\nu)$ denote the distance from $p$ to the endvertex
$\nu$ along the edge (curve) $e_x$.
Then, without loss of generality, 
$d_{e_1}(p,u)\le d_{e_1}(p,v)$ and
$d_{e_2}(p,w)\le d_{e_2}(p,z)$.
Note that, by the definition of a natural disk neighborhood system,
the radius of $u$'s disk is at least $d_{e_1}(p,u)$
and the radius of $w$'s disk is at least $d_{e_2}(p,w)$;
hence, $u$'s disk intersects $w$'s disk.
Let us therefore charge this pair of vertices for this edge crossing.
Note that, since $G$ has bounded degree, this pair of vertices
can be charged at most $O(1)$ times in this way (specifically, 
each such pair can be charged at most $O({\rm degree}(G)^2)$ times).
Furthermore,
by Lemma~\ref{lemma-pairs}, there are only $O(n)$ pairs of intersecting disks
in $S$.
Therefore, the number of pairs of edges of $G$ that cross in $\R^2$ is $O(n)$.
\end{proof}

Thus, our characterization allows for an even larger number of
intersections than occur in practice. More importantly, our
characterization allows for us to construct efficient algorithms for
important problems involving road networks.

\section{Separator Decomposition and Its Applications}
In some applications, we desire 
algorithms for small-sized \emph{separators}.
Given a graph $G$, an
{\it $f(n)$-separator} for $G$ is a set of $f(n)$ vertices whose removal
from $G$ results in $G$ being subdivided into two subgraphs having
at most $\delta n$ vertices each~\cite{g-psppt-95,lt-stpg-79,m-fsscs-86},
for some fixed constant $0<\delta<1$.

We may alternatively desire
algorithms for small-sized \emph{geometric separators}, which are
defined for disk systems.
Given a collection of disks $S$ in $\R^2$,
a \emph{Jordan curve} $f(n)$-separator 
for $S$ is a Jordan curve $J$
that
intersects at most $f(n)$ disks in $S$ so that at most $\delta n$
disks are
either inside the interior or exterior of $J$, for some fixed
constant
$0<\delta<1$.
Ideally, $J$ should be a simple curve, like a circle, and indeed,
Eppstein {\it et al.}~\cite{emt-dltag-93} 
show that $k$-ply disk neighborhood
systems have linear-time computable
circle $O(\sqrt{n})$-separators for $\delta=3/4$, where $k$ is a
constant.
We immediately have the following:

\begin{theorem}
\label{thm-deterministic}
Suppose we are given a set $S$ that is
an $O(\sqrt{n})$-exceptional $O(1)$-ply disk neighborhood system.
Then one can deterministically construct a geometric
$O(\sqrt{n})$-separator decomposition of
$S$, for $\delta=3/4$, 
in $O(n)$ time.
\end{theorem}
\begin{proof}
The algorithm of Eppstein {\it et al.}~\cite{emt-dltag-93} 
constructs a separator decomposition assuming the disk system
has constant ply. Their algorithm can be extended to
an $O(\sqrt{n})$-exceptional $O(1)$-ply disk neighborhood system,
however. Note that removing all the exceptional disks from $S$ results
in a constant-ply disk system, $S'$, and adding all the exceptional disks to 
an $O(\sqrt{n})$-sized separator for $S'$ will result in 
an $O(\sqrt{n})$-sized separator for $S$, whose size is at most
double that for $S'$.
Moreover, the algorithm of
Eppstein {\it et al.}~\cite{emt-dltag-93} will still work in linear
time for $S$, as it is based on a search in dual space for a point
corresponding to the cutting circle with respect to a desired
separator size. This point can still be
found in linear time for $S$ using their algorithm, just by
increasing the desired size of the separator by a factor of $2$.
\end{proof}

We also have the following randomized result for constructing an
entire recursive separator decomposition for a road network, $G$.

\begin{theorem}
\label{thm-randomized}
Suppose we are given an $n$-vertex geometric graph $G$
having at most $O(n/\log^{(c)} n)$ edge crossings,
for some constant $c$.
Then we can construct a recursive $O(\sqrt{n})$-separator
decomposition of
$G$ in $O(n)$ expected time, for $\delta=2/3$.
\end{theorem}
\begin{proof}
The proof follows from Theorems 4.1, 5.1, and 6.1 of
Eppstein {\it et al.}~\cite{egs-ltagg-09}, specialized to 
an $n$-vertex geometric graph $G$
having at most $O(n/\log^{(c)} n)$ edge crossings.
\end{proof}

\subsection{Applications}
In this section, we explore some applications of our linear-time
separator decomposition algorithms for road networks.
Given an $n$-vertex
bounded-degree graph $G$ and a recursive $O(\sqrt{n})$-separator 
decomposition for $G$,
Henzinger {\it et al.}~\cite{hkrs-fspap-97} show that one can compute shortest
paths from a single source $s$ in $G$ to all other vertices in $G$ in $O(n)$
time.
Using the separator decomposition algorithm presented above
in Theorem~\ref{thm-randomized},
we can show that
their algorithm applies to road networks.

Suppose, then, we are given $k$ distinguished vertices in an $n$-vertex
road network, $G$, that is, a multiscale-dispersed graph, and we
wish to construct 
the \emph{Voronoi diagram} of $G$, which is a labeling of each
vertex $v$ of $G$ with the name of the distinguished vertex closest
to $v$.
As before, we can assume without loss of generality 
that $G$ has constant degree.
In this case, we construct a recursive $O(\sqrt{n})$-separator
decomposition of $G$ using one of the algorithms of the previous
subsection.
Let $B$ be the recursion tree and let us label each vertex $v$ in $G$
with the internal node $w$ in $B$ where $v$'s disk appears as a separator
or with the leaf $w$ in $B$ corresponding to a set containing $v$ where
we stopped the recursion (because the set's size was below our
stopping threshold).
Given this labeling, we can trace out the subtree $B'$ of $B$ 
that consists of the union of paths from the root of $B$ to the
distinguished nodes in $G$ in $O(n)$ time.
Let us now assign each edge in $B'$ to have weight $0$ and let us add
$B'$ to $G$ to create a larger graph $G'$.
Note that if we add each internal node $v$ in $B'$ to the separator
associated with node $v$ in $B$, then we get a recursive 
$O(\sqrt{n})$-separator decomposition for $G'$, for each separator in
the original decomposition increases by at most one vertex.
Thus, we can apply the algorithm of Henzinger {\it et al.}~\cite{hkrs-fspap-97}
to compute the shortest paths in $G'$ from the root of $B'$ to every
other vertex in $G'$ in $O(n)$ time.
Moreover, since the edges of $G'$ corresponding to edges of $B'$ have
weight $0$, this shortest path computation will give us the Voronoi
diagram for $G$.
Therefore, we have the following.

\begin{theorem}
\label{thm-vd}
Given a road network, modeled as 
a connected $n$-vertex multiscale-dispersed graph, $G$, 
having $O(n/\log^{(c)} n)$ edge crossings,
for some constant $c$,
we can compute shortest paths from any vertex
$s$ or the Voronoi diagram defined by any set of $k$ vertices in $G$
in $O(n)$ expected time.
\end{theorem}

There are also a number of other applications of our algorithmic
approach for studying road networks.

\section{Circle Arrangements}
Suppose we are given an $n$-vertex
multiscale-dispersed graph $G=(V,E)$, embedded in the plane, and 
we wish to construct an explicit
embedding of the intersection graph, $H$,
for $G$'s natural disk neighborhood system, $S$.
We assume that 
$G$ is connected and has bounded degree, since path problems only
make sense in connected road networks and
real-world road junctions 
cannot have an arbitrary number of incoming roads.
Furthermore, we assume $G$ is specified in terms of its embedding $G'$ on the
sphere, with the edge crossings explicitly represented.
That is, $G'$ is a plane graph whose vertices are the vertices and edge
crossings of $G$, which is commonly referred to as the \emph{planarization}
of $G$. The existence of such a specification of $G$ is dervided from
the result
of Eppstein {\it et al.}~\cite{egs-ltagg-09}, 
which shows that such a representation can be constructed in $O(n)$
time as long as there are at most $O(n/\log^{(c)} n)$ edge crossings,
for some constant $c$.

\subsection{The $k$-Neighborly Property}
Define a \emph{grid shortcut} for a vertex $v$ in $G'$ to be 
a (new) edge from $v$ to one of the endpoints of the first 
edge (non-incident to $v$) in $G$ that is
hit by a vertical or horizontal ray emanating from $v$.
Likewise, define a \emph{grid augmentation} of $G$ to be a graph
${\hat G}$ such that ${\hat G}$ includes each vertex and edge from $G$
plus every (directed) grid shortcut edge emanating from a vertex in $G$.
That is,
${\hat G}$ is a mixed graph~\cite{gt-adfai-02}, containing both directed and
undirected edges, so that, for example, the out-edges of each
vertex $v$ includes both the original undirected edges incident to $v$ 
in $G$ and the new directed edges with origin $v$ in ${\hat G}$.
Note that if $G$ has bounded degree,
then the out-degree of each vertex in the grid augmentation ${\hat G}$ 
is also bounded (in particular, the out-degree of 
each vertex is increased by at most $4$).

Given a multiscale-dispersed graph, 
$G$, together with its natural disk neighborhood
system $S$ and grid augmentation ${\hat G}$, we say that
$G$ is \emph{$k$-neighborly} if, for every pair of intersecting disks
in $S$, centered respectively at vertices $v$ and $w$ in $G$, 
there is a $k_1$-edge path $P_1$ from $v$ to $w$ in ${\hat G}$ and a 
$k_2$-edge path $P_2$
from $w$ to $v$ in ${\hat G}$ such that $k_1\le k$ and $k_2\le k$.
Intuitively, this notion of a $k$-neighborly network 
formalizes the property that if two vertices $v$ and $w$ 
are close enough in $G$ so that their associated disks intersect, 
then there should be a set of driving directions 
of length at most $k$
for moving from $v$ to $w$, possibly using vertical or horizontal shortcuts.
Interestingly, our experimental analysis supports the claim that 
real-world road 
networks are $O(1)$-neighborly (see Figure~\ref{fig-dista}), 
even though it is not the case that all pairs
of nearby vertices in real-world
road networks have constant-sized sets of driving
directions that avoid the use of vertical or horizontal shortcuts
(see Figure~\ref{fig-distb}).
Both plots are taken from distance experiments on the same sample of 
road networks in the TIGER/Line database.

\begin{figure}[htb]
\begin{center}
\includegraphics[width=3.5in]{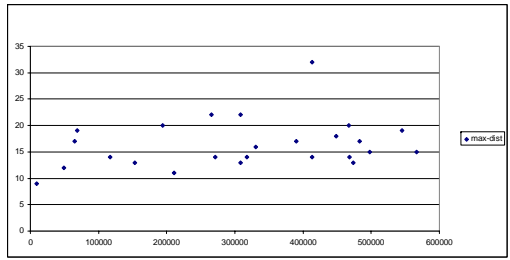} \\
\end{center}
\vspace*{-1in}
\caption{
A scatter plot of the maximum edge distance between centers of two
intersecting disks in a natural disk neighborhood system with the search
expanded to include edges from a grid augmentation of the original road
network.}
\label{fig-dista}
\end{figure}

\begin{figure}[hbt]
\begin{center}
\includegraphics[width=3.5in]{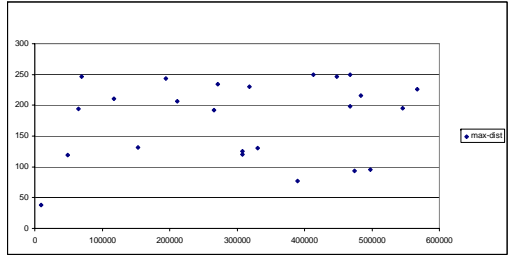} 
\end{center}
\vspace*{-1in}
\caption{
A scatter plot of the maximum edge distance between 
centers of two
intersecting disks in a natural disk neighborhood system (distance
searches where artificially cutoff after a depth of 250, so the
distances could be even higher), without grid augmentations.
}
\label{fig-distb}
\end{figure}

Let us assume, then, that we are given an $O(1)$-neighborly, bounded-degree,
connected multiscale-dispersed graph, $G$, 
and we are interested in constructing an
explicit representation of the intersection graph, $H$, for $G$'s natural
disk neighborhood system.
Furthermore,
let us assume that we are given a list of all pairs of edges in
$G$ that cross in the embedding of $G$ in $\R^2$,
which converts $G$ into a plane graph, $G'$.
Moreover, since $G$ is a multiscale-dispersed graph, we have the following
useful combinatorial lemma.

Suppose we are given the planarization
$G'$ of $G$.
Thus, the plane graph $G'$ has $O(n)$ vertices, which, in turn, implies that
$G'$ has $O(n)$ edges. Moreover, $G'$ is connected,
which implies that each face of $G'$ can be viewed as a simple polygon, if we
consider each edge as having a left side and a right side.
Thus, we can compute a vertical trapezoidal decomposition of each face in $G'$
in $O(n)$ time~\cite{c-tsplt-91a,agr-ratsp-01} and 
we can similarly compute a horizontal trapezoidal decomposition 
of each face in $G'$.
By scanning these two trapezoidal decompositions, we can compute
the grid augmentation, ${\hat G}$, of $G$.
Since $G$ is $k$-neighborly, for some constant $k$, 
and $G$ has bounded degree, given any vertex $v$, we can 
find all the neighbors of $v$'s disk in $S$ by searching the $O(1)$ vertices
in ${\hat G}$ that are at most $k$ edges away from $v$ in ${\hat G}$.
Therefore, we have the following.

\begin{lemma}
\label{lemma-neighborly}
Given an $n$-vertex $O(1)$-neighborly,
bounded-degree, connected multiscale-dispersed graph, $G$, together
with its planarization, $G$, we can construct an explicit
representation of the intersection graph, $H$, of $G$'s natural disk
neighborhood system, $S$, in $O(n)$ time.
\end{lemma} 

\subsection{Inductive $k$-Clustering}
Suppose we are given a multiscale-dispersed graph, $G$, together with
the disk intersection graph, $H$,
for $G$'s natural disk neighborhood system, $S$.
In particular, let us assume we have
an explicit representation of $H$, 
such as in an adjacency list,
with each vertex in $H$ labeled with its representative disk from $S$.
For each vertex $v$ in $H$, let $R(v)$ denote the set of vertices in $H$ that
are adjacent to $v$ and whose disk has radius no bigger that $v$'s disk.
We say that $G$ is \emph{inductively $k$-clustered} if,
for each vertex $v$ in $H$, the number 
of connected components in $R(v)$ is no more than $k$.
In the experimental section, we provide empirical evidence that
real-world road networks are inductively $O(1)$-clustered (although the
intersection graphs of their natural disk neighborhood systems do
not necessarily have bounded degree).
So, suppose that our given intersection graph, $H$,
is for the natural disk neighborhood system, $S$, of
an inductively $O(1)$-clustered multiscale-dispersed graph, $G$.

This characterization is justified by additional empirical evidence.
Indeed,
Figure~\ref{fig-degb} shows a scatter plot that supports the
claim that real-world road networks are inductively $O(1)$-clustered, that is, the number
of connected components of smaller disks neighboring a given disk is
$O(1)$.

\begin{figure}[htb]
\vspace*{-0.2in}
\begin{center}
\includegraphics[width=3.75in]{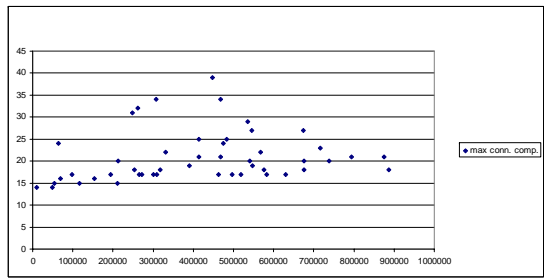} 
\end{center}
\vspace*{-1in}
\caption{
A scatter plot of the maximum number of connected components among smaller
disks, which gives evidence that each road network is inductively
$O(1)$-clustered.}
\label{fig-degb}
\end{figure}

By the way, an alternate
approach for designing a linear-time algorithm for constructing an arrangement from an
explicit representation of a disk network system is not supported by
the data.
Namely, it would be great if the vertex 
degrees in the disk neighborhood system were constant.
Unfortunately, as shown in Figure~\ref{fig-dega}, the distribution of maximum
disk-intersection degrees does not appear to be constant.
Indeed, it appears to be proportional to $\sqrt{n}$, which supports
the ``$O(\sqrt{n})$-exceptional'' part of the notion of road networks as 
being subgraphs of $O(\sqrt{n})$-exceptional $O(1)$-ply disk systems.

\begin{figure}[hbt]
\vspace*{-0.2in}
\begin{center}
\includegraphics[width=3.5in]{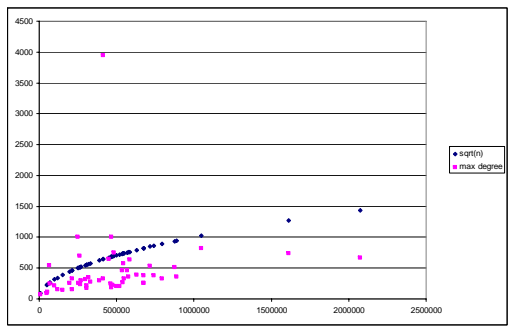} 
\end{center}
\vspace*{-0.75in}
\caption{
A scatter plot of the maximum vertex degrees for vertices
in the natural disk neighborhood systems  for TIGER/Line road
networks, indexed by network size (the outlier is Louisiana),
together with values for $\sqrt{n}$.
}
\label{fig-dega}
\end{figure}

The inductive $k$-clustering property
allows us to design a linear-time arrangement construction
algorithm for $S$.
Direct each edge in $H$ to go from $v$ to $w$ if the disk for $v$ is
smaller than the disk for $w$, with ties broken using the indices of $v$
and $w$, so that the resulting graph $D$ is a directed acyclic graph.
Note that, by Lemma~\ref{lemma-pairs},
$H$ has $O(n)$ edges; hence, $D$ has $O(n)$ edges.
Let us therefore perform a topological ordering of the vertices in $D$,
which will take $O(n)$ time.
Now, consider each vertex $v$
according to this topological ordering.
In processing a vertex $v$ in this order, we may
assume inductively that we have constructed the arrangement of all the
circles associated with vertices in $R(v)$.
Since there are $O(1)$ connected components in $R(v)$, each with its own
subarrangement, we may sort these components radially, using the center of
$v$'s circle as origin, in $O(1)$ time.
Note that these connected components don't intersect each other (by
definition).
So the arrangement of the circles in $R(v)$ and $v$'s circle can be
constructed in $O(|R(v)|)$ time by ``splicing'' $v$'s circle into this
arrangement using a traversal through each subarrangement (with the
subarrangements considered in turn according to their
radial order around $v$'s circle center).
This gives us the following.

\begin{lemma}
Suppose we are
given an $n$-vertex inductively $O(1)$-clustered 
multiscale-dispersed
graph, $G$, together with the intersection graph, $H$,
of $G$'s natural disk neighborhood system, $S$.
Then we can construct the arrangement of the
circles in $S$ in $O(n)$ time.
\end{lemma}

This lemma, together with Lemma~\ref{lemma-neighborly},
imply the following:

\begin{theorem}
Suppose we are given a bounded-degree $n$-vertex 
connected multiscale-dispersed graph, $G$, 
such that $G$ is $O(1)$-neighborly
and inductively $O(1)$-clustered.
Suppose further that we are given
a planarization, $G'$, of $G$.
Then we can construct the arrangement of the
circles in $G$'s natural disk neighborhood system, $S$, in $O(n)$ time.
\end{theorem}

\section{Conclusions and Future Work}
We have given a new theoretical characterization of 
road networks as multiscale-dispersed graphs and we have shown how this
characterization leads to algorithms that provably run in linear
time for a number of interesting problems.
Our main techniques involve the construction of recursive
$O(\sqrt{n})$-separator decompositions and the applications of such
decompositions to Voronoi diagrams and shortest paths. 
In addition, 
we provided an algorithm for constructing the arrangement of the
natural disk neighborhood system of a road network, utilization some
additional properties of road networks, which are justified
empirically.

There are a number of interesting open problems and future research directions 
raised by this paper, including: 
\begin{itemize}
\item
Can one construct 
a recursive circle $O(\sqrt{n})$-separator decomposition of a
multiscale-dispersed graph deterministically in $O(n)$ time?
\item
Can one construct a trapezoidal decomposition of an $n$-vertex 
multiscale-dispersed graph 
in $O(n)$ time? (This is related to the well-known open problem of
computing a trapezoidal decomposition of an $n$-vertex non-simple polygon in
$O(n+k)$ time, where $k$ is the number of its edge crossings.)
\item
Are there a set of road network complexity measures, possibly based on
our definitions, that can replace the Alpha and Gamma indices~\cite{c-esag-96},
which are used in the transportation network literature and
are based on the flawed notion that road networks are plane graphs?
\end{itemize}
Of course, one can also perform additional experiments to empirically test
whether real-world road networks possess additional properties not
considered in this paper.

\subsection*{Acknowledgments}
This work was supported in part by the NSF, under grant 0830403,
and by the Office
of Naval Research, under grant N00014-08-1-1015.
With the exception of the images in Figure~\ref{fig-fractal}, 
all figures in this
paper are the property of the authors and are used by permission.

\begin{flushleft}
\bibliographystyle{abbrv}
\bibliography{extra,geom,../roads,roads2,goodrich}
\end{flushleft}

\end{document}